\documentclass[letterpaper, 10 pt, conference]{ieeeconf}

\IEEEoverridecommandlockouts
\usepackage[utf8]{inputenc} 
\usepackage[T1]{fontenc}
\usepackage{url}
\usepackage[compress]{cite}
\usepackage{amsmath}
\usepackage{graphicx}
\usepackage{booktabs}
\usepackage{physics}
\usepackage{amssymb}
\usepackage{tikz}
\usepackage{framed}

\newtheorem{theorem}{\bf Theorem}
\newtheorem{proposition}{\bf Proposition}

\newtheorem{corollary}{\bf Corollary}

\newtheorem{lemma}{\bf Lemma}
\newtheorem{definition}{\bf Definition}

\makeatletter
\renewcommand*{\@opargbegintheorem}[3]{\trivlist
  \item[\hskip \labelsep{\bfseries #1\ #2}] \textbf{(#3)}\ \itshape}
\makeatother

\begin{document}

\title{Privacy Against Hypothesis-Testing Adversaries for Quantum Computing} 

\author{Farhad Farokhi}

\maketitle

\thispagestyle{plain}
\pagestyle{plain}

\begin{abstract}
 A novel definition for data privacy in quantum computing based on quantum hypothesis testing is presented in this paper. The parameters in this privacy notion possess an operational interpretation based on the success/failure of an omnipotent adversary being able to distinguish the private categories to which the data belongs using arbitrary measurements on quantum states. Important properties of post processing and composition are then proved for the new notion of privacy. The relationship between privacy against  hypothesis-testing adversaries, defined in this paper, and quantum differential privacy are then examined. It is shown that these definitions are intertwined in some parameter regimes. This enables us to provide an interpretation for the privacy budget in quantum differential privacy based on its relationship with privacy against hypothesis testing adversaries.
\end{abstract}

\section{Introduction}
Quantum computing algorithms have garnered huge attention due to their considerable speedups in several classically difficult problems, such as factorising~\cite{shor1994algorithms}. These breakthroughs and the added attention has paved the way for development of new algorithms for big data processing, such as quantum machine learning~\cite{lloyd2014quantum, huang2021information,biamonte2017quantum}. However, data processing can result in unintended information leakage~\cite{kearns2019ethical}. This is an important issue because, as quantum hardware becomes more commercially available, these algorithms can be implemented on real-world sensitive, private, or proprietary datasets. Therefore, there is a need for development of frameworks to better understand private information leakage in quantum computing algorithms and to construct privacy-preserving algorithms.

In the classical computing literature, differential privacy has become the gold standard of privacy analysis and private algorithm design~\cite{dwork2006calibrating,dwork2008differential,abadi2016deep}. This is often attributed to the fact that differential privacy makes minimal assumptions about the data (e.g., range rather than distribution) and meets important properties of post processing and compositions~\cite{dwork2014algorithmic}. Although possessing powerful guarantees, differential privacy has been polarizing~\cite{hotz2022balancing,bhaskar2011noiseless,nabar2006towards}. Criticisms surroundings conservativeness of differential privacy have motivated a host of studies on privacy in information theory that can handle privacy-utility trade-off better in certain situations~\cite{issa2019operational,farokhi2021measuring,liao2019tunable, li2018privacy,farokhi2017fisher}. In fact, adoption of hypothesis-testing and estimation-based adversaries have been proposed as less conservative alternatives to differential privacy by social scientists following implementation of differential privacy in the 2020 United States Decennial Census of Population and Housing~\cite{hotz2022balancing}. Nonetheless, differential privacy has been recently extended to quantum computing~\cite{zhou2017differential,aaronson2019gentle,hirche2022quantum}. However, very little attention has been paid other forms of privacy in quantum systems. In this paper, we investigate privacy against hypothesis-testing adversaries. This is of particular interest to us due to the need for providing an operational, intuitive notion of privacy with real-world interpretations for privacy analysis and guarantees, which is somewhat absent in the differential privacy literature. 

In this paper, we particularly propose a novel definition for data privacy for quantum computing based on quantum hypothesis testing. The design parameters in this notion of privacy possess an operational interpretation (specifically for general lay-users) based on the success/failure of an omnipotent adversary being able to distinguish the private class to which the data belongs (e.g., suffering from a certain disease in health datasets or belonging to the training dataset in membership inference attacks) based on the arbitrary measurement operators. We prove two very important properties for the new notion of privacy: post processing and composition. These properties are highly sought-after in privacy definitions~\cite{zhou2017differential} and information leakage metrics~\cite{issa2019operational}. Subsequently, we investigate the relationship between privacy against hypothesis-testing adversaries and quantum differential privacy. This enables us to provide an interpretation for parameters of differential privacy based on its relationship with privacy against hypothesis-testing adversaries in certain parameter regimes. We will finally investigate the effectiveness of differential privacy against hypothesis-testing adversaries. 

The remainder of this paper is organized as follows. We provide a review of basic concepts in quantum computing and information in Section~\ref{sec:review}. The definition and results on privacy against hypothesis-testing adversaries is presented in Section~\ref{sec:privacy_hypothesis}. Section~\ref{sec:diff_privacy} presents quantum differential privacy and its relationship with privacy against hypothesis-testing adversaries. Finally, we present some concluding remarks and future directions for research in Section~\ref{sec:discussions}. 

\section{Quantum states and channels} \label{sec:review}
The definitions and preliminary results in this review section are mostly borrowed from~\cite{wilde2013quantum}. When the results or definitions are from outside this source, appropriate citations are presented.

A quantum system is modelled by a Hilbert space $\mathcal{H}$, i.e., a complex vector space, equipped with an inner product, that is complete with respect to the norm defined by the inner product. Throughout this paper, Dirac's notation is used to denote quantum states. That is, a \textit{pure quantum state}, which is an element (i.e., vector) of Hilbert space $\mathcal{H}$ with unit norm, is denoted by `ket' $\ket{\cdot}$, e.g., $\ket{\psi}\in\mathcal{H}$. The inner product of two states $\ket{\phi}$ and $\ket{\psi}$ is denoted by $\braket{\phi}{\psi}$. Here, `bra' $\bra{\psi}$ is used to refer to conjugate transpose of $\ket{\psi}$ and $\braket{\phi}{\psi}:=\bra{\phi}\,\ket{\psi}\in\mathbb{C}$. 

The basic element of interest in quantum information theory is a quantum bit, which is often referred to as the \textit{qubit}. A qubit is a 2-dimensional quantum state. Any qubit can be written in terms of the so-called computational basis $\ket{0}$ and $\ket{1}$ that form an orthonormal basis for the two dimensional Hilbert space modelling the qubit, that is, any qubit can be written as $\ket{\psi}=\alpha\ket{0}+\beta\ket{1}$ with $\alpha,\beta\in\mathbb{C}$ such that $|\alpha|^2+|\beta|^2=1$. Combination of two qubits $\ket{\phi}$ and $\ket{\psi}$ is denoted by their tensor product $\ket{\phi}\otimes\ket{\psi}$, where $\otimes$ is the Kronecker or tensor product. For the sake of brevity, we sometimes refer to $\ket{\phi}\otimes\ket{\psi}$ as $\ket{\phi}\ket{\psi}$ or $\ket{\phi,\psi}$. When two qubits $\ket{\phi}$ and $\ket{\psi}$ belong to or assigned to two distinct registers $A$ and $B$ (e.g., qubits used by two separate parties), and this information is either unclear from the context or must be emphasized, we write $\ket{\phi}_A\otimes\ket{\psi}_B$ or $\ket{\phi}_A\ket{\psi}_B$. A quantum (logic) gate is any unitary operator, e.g., $U$ such that $U^\dag U=I$, that acts on a quantum state. Note that, here, $U^\dag$ denotes the conjugate transpose of $U$. 

A \textit{mixed quantum state} is represented by an ensemble $\{(p_1,\ket{\psi_1}),\dots,(p_k,\ket{\psi_k})\}$ such that $p_i\geq 0$ for all $i\in[k]:=\{1,\dots,k\}$ and $\sum_{i\in[k]}p_i=1$. A mixed quantum state implies that the quantum system is in pure state $\ket{\psi_i}$ with probability $p_i$ for all $i\in[k]$. A convenient way to model and analyse mixed quantum state is to use density operators. The density operator corresponding to ensemble $\{(p_1,\ket{\psi_1}),\dots,(p_k,\ket{\psi_k})\}$ is given by $\rho:=\sum_{i\in[k]} p_i\ket{\psi_i}\bra{\psi_k}$. Evidently, by construction, $\trace(\rho)=1$. Note that pure quantum states $\ket{\phi}$ can also be modelled using rank-one density operator $\rho=\ket{\phi}\bra{\phi}$. Therefore, there is no loss of generality to work with density operators even when dealing with pure quantum states. Combination of two density operators $\rho$ and $\sigma$ is denoted by their tensor product $\rho\otimes\sigma$.

A basic operation in quantum systems is measurement, which enables extraction of information about the quantum states of the systems. A measurement is modelled by a set of operators $M=\{K_i\}_{i\in[m]}$ with normalization constraint that $\sum_{i\in[m]}K_i^\dag K_i=I$. By performing measurement $M$ on a quantum system with state $\rho$, we observe output $i\in[m]$ with probability $\trace(K_i\rho K_i^\dag)$ in which case, after the measurement, the state of the quantum system is $K_i\rho K_i^\dag/\trace(K_i\rho K_i^\dag)$. When the post-measurement state of the quantum system is of no interest, we can use the positive operator-valued measure (POVM) framework, which is a set of positive semi-definite Hermitian matrices $F=\{F_i\}_{i\in[m]}$ such that $\sum_{i\in[m]} F_i=I$. In this case, the probability of obtaining output $i\in[m]$ when taking a measurement on a system with quantum state $\rho$ is given by $\trace(\rho F_i)=\trace( F_i \rho)$. 

A quantum channel is the most general quantum operation. A quantum channel is a mapping from the space of density operators to potentially another space of density operators that is both completely positive and trace preserving. Quantum channels model open quantum systems, i.e., quantum systems that interact with environment, and thus can model noisy quantum behaviours. According to Choi-Kraus theorem~\cite[Theorem~4.4.1]{wilde2013quantum}, for each quantum channel $\mathcal{E}$, there exists a family of linear operators $\{E_j\}_{j\in[n]}$ for some $n\in\mathbb{N}$ such that $\sum_{j}E_j^\dag E_j=I$ and $\mathcal{E}(\rho)=\sum_{j\in [n]} E_j\rho E_j^\dag$ for all density operators $\rho$. This is referred to as the Kraus representation of quantum channels. For instance, a quantum (logic) gate with unitary operator $U$ can be represented by $\mathcal{E}(\rho)=U\rho U^\dag$. Similarly, if we discard or delete the outcome of measurement $M=\{K_i\}_{i\in[m]}$, the quantum state transition can be modelled by quantum channel $\mathcal{E}(\rho)=\sum_{i\in[k]} K_i\rho K_i^\dag$. We define the tensor product of quantum channels $\mathcal{E}_1$ and $\mathcal{E}_2$ as $\mathcal{E}_1\otimes\mathcal{E}_2(\rho_1\otimes\rho_2):=\mathcal{E}_1(\rho_1)\otimes\mathcal{E}_2(\rho_2)$ for all density operators $\rho_1$ and $\rho_2$. 

The trace norm or Schatten 1-norm of any linear operator $M$ is defined as $\|M\|_1:=\trace(|M|)=\trace(\sqrt{M^\dag M})$. Based on this, we can define the trace distance between any two density operators $\rho$ and $\sigma$ with 
$\mathcal{T}(\rho,\sigma):=\frac{1}{2}\|\rho-\sigma\|_1\in[0,1].$ Recall that density operators belong to the set of linear operators (i.e., matrices). 
The distance is equal to zero when two quantum states are equal. However, the distance attains its maximum value when two quantum states have support on orthogonal subspaces. For $\upsilon \in [0,1]$, the $\upsilon$-relative entropy between two quantum states $\rho$ and $\sigma$ is defined as
$D^{\upsilon }(\rho \|\sigma )\!=\!-\!\log\left(\min\{\trace(Q\sigma) |0\!\preceq \!Q\!\preceq \!I, \trace(Q\rho) \geq 1\!-\!\upsilon \}\right).$
The $\upsilon$-relative entropy satisfies a few important properties that we will use in this paper. These properties are borrowed from~\cite{wang2012one}. First, $D^{\upsilon }(\rho \|\sigma )\geq 0$ with equality if $\rho =\sigma$ and $\upsilon=0$. Second, $\upsilon$-relative entropy enjoys data processing inequality, i.e., $D^{\upsilon }(\mathcal{E}(\rho) \|\mathcal{E}(\sigma) )\leq D^{\upsilon }(\rho \|\sigma )$ for all density operators $\rho,\sigma$ and all quantum channels $\mathcal{E}$. Also, $D^{\upsilon }(\rho \|\sigma )\leq ({S(\rho \| \sigma)+H_b(\upsilon)})/({1-\upsilon}),$
where $H_b(\upsilon)=-\upsilon\log(\upsilon)-(1-\upsilon)\log(1-\upsilon)$ is the binary entropy function and $S(\rho\|\sigma):=\trace(\rho(\log(\rho)-\log(\sigma)))$ is the usual relative entropy in quantum information theory. The $\upsilon$-relative entropy and the trace distance also satisfy the following relationship $\nu/(1-\nu)\|\rho-\sigma\|_1\leq D^{\upsilon }(\rho \|\sigma )$~\cite{dupuis2014generalized}. The smooth max-relative entropy is defined as $D^\upsilon_{\max}(\rho\|\sigma)=\inf_{\tau\in\mathcal{B}^\upsilon(\rho)} D_{\max}(\tau\|\sigma)$, where $D_{\max}(\tau\|\sigma)=\inf\{\lambda\geq 0\,|\,\rho\preceq\exp(\lambda)\sigma \}$ and $\mathcal{B}^\upsilon(\rho):=\{\tau\,|\,\tau^\dag=\tau\succeq 0, \|\rho-\tau\|_1\leq 2\upsilon\}$. 

Depolarizing channel is an important type of quantum noise that is represented by
\begin{align} \label{eqn:dep_channel}
    \mathcal{E}_{\rm Dep}(\rho):=\frac{p}{D}I+(1-p)\rho,
\end{align}
where $D$ is the dimension of the Hilbert space to which the system belongs and $p\in[0,1]$ is a probability parameter.

\section{Quantum Hypothesis Testing and Privacy} \label{sec:privacy_hypothesis}
Consider a quantum hypothesis testing scenario where a decision maker aims to distinguish between two quantum states $\rho$ (null hypothesis) and $\sigma$ (alternative hypothesis). This is done by performing POVM $M:=\{M_1,M_2\}$ with $M_1+M_2=I$ and $0\preceq M_i\preceq I$ for $i=1,2$. If measurement outcome corresponding to the operator $M_1$ is realized, the decision maker guesses that the state is $\rho$ while, if measurement outcome corresponding to the operator $M_2$ is realized, the decision maker guesses that the state is $\sigma$. The probability of a type-I error (false positive) is equal
\begin{align} \label{eqn:false_positive}
    \alpha(M_2):=\trace(M_2\rho).
\end{align}
The probability of a type-II error (false negative) is given by
\begin{align}\label{eqn:false_negative}
    \beta(M_1):=\trace(M_1\sigma).
\end{align}
The optimal test,  which seeks to minimize the false negative probability subject to a constraint on maintaining the false positive probability below $\eta\in[0,1]$, is given by
\begin{subequations}
\begin{align}
    \beta_\eta(\rho,\sigma):=\min_{M_1,M_2\succeq 0} &\; \beta(M_1),\\
    \mathrm{s.t.} \quad &\; M_1+M_2=I,\\
                    &\; \alpha(M_2)\leq \eta.
\end{align}
\end{subequations}
This is referred to as \textit{asymmetric quantum hypothesis testing}~\cite{regula2022postselected}. The following well-known result (see, e.g.,~\cite{wang2012one}) can be easily derived based on the definition of $\beta_\eta(\rho,\sigma)$ and $\eta$-relative entropy $D^{\eta }(\rho \|\sigma )$. 

\begin{proposition} \label{prop:assym}
$\beta_\eta(\rho,\sigma)=2^{-D^{\eta }(\rho \|\sigma )}.$
\end{proposition}

\begin{proof}
    Note that
\begin{align*}
     \beta_\eta(\rho,\sigma) \!
        =&\min_{M_1,M_2\succeq 0} \{\trace(M_1\sigma)|
            M_1\!+\!M_2\!=\!I,\trace(M_2\rho)\!\leq\! \eta\}\\
        =&\min_{I \succeq M_1\succeq 0} \{\trace(M_1\sigma)|
            \trace((I-M_1)\rho)\leq \eta\}\\
        =&\min_{I \succeq M_1\succeq 0} \{\trace(M_1\sigma)|
            1-\eta\leq \trace(M_1\rho)\}\\
        =&2^{-D^{\eta }(\rho \|\sigma )}.
\end{align*}
This concludes the proof.
\end{proof}

Alternatively, a combination of false positive and false negative probabilities can be minimized:
\begin{align*}
    p_{\rm err}(\rho,\sigma) :=\min_{M_1,M_2\succeq 0} &\; p_\rho \alpha(M_2) +p_\sigma \beta(M_1),\\
    \mathrm{s.t.} \quad &\; M_1+M_2=I,
\end{align*}
where $p_\rho\in[0,1]$ and $p_\sigma\in[0,1]$, respectively, denote the prior probability that quantum state $\rho$ and the prior probability that quantum state $\sigma$ are prepared. Clearly, by construct, $p_\rho+p_\sigma=1$. This is referred to as \textit{symmetric quantum hypothesis testing}~\cite{regula2022postselected}. 

\begin{theorem}[{{Helstrom-Holevo theorem~\cite[p.~254-255]{wilde2013quantum}}}]
$p_{\rm err}(\rho,\sigma)=\frac{1}{2}\left(1-\|p_\rho\rho-p_\sigma\sigma\|_1\right).$
\end{theorem}

The most indistinguishable quantum states are $\rho=\sigma$. In this case, a decision maker would not be able to identify the quantum states \textit{because their observable are equivalent}. Therefore, we can define 
\begin{align*}
    p_{\rm max}
    :=&p_{\rm err}(\rho,\rho)\\
    =&\frac{1}{2}\left(1-\|(p_\rho-p_\sigma)\rho\|_1\right)\\
    =&\frac{1}{2}\left(1-|p_\rho-p_\sigma|\|\rho\|_1\right)\\
    =&\frac{1}{2}\left(1-|p_\rho-p_\sigma|\right).
\end{align*}
Therefore, for general density operators, we have
\begin{align*}
    p_{\rm err}(\rho,\sigma)=p_{\rm max}+\frac{1}{2}\left(|p_\rho-p_\sigma|-\|p_\rho\rho-p_\sigma\sigma\|_1\right).
\end{align*}

In quantum data privacy, it is desired to protect the quantum state of a system (which is being used for quantum computation) from being accurately estimated. Particularly, given a quantum state $\rho$, we want to make sure that no decision maker can identify whether the quantum state of the system is $\rho$ or another \textit{similar} quantum state $\sigma$. Similarity is modelled or captured using the neighbourhood relationship, c.f., differential privacy~\cite{hirche2022quantum}.

\begin{definition}[Neighbouring Relationship]
A neighbouring or similarity relationship over the set of density operators is a mathematical relation that is both reflective and symmetric. The notation $\rho\sim\sigma$ signifies that two quantum states $\rho$ and $\sigma$ are neighbouring or similar. Note that, by definition, $\rho\sim\rho$ (reflectivity) and $\rho\sim\sigma$ implies $\sigma\sim\rho$ (symmetry). 
\end{definition}

An example of neighbouring or similarity relationship is the notion defined using trace distance in~\cite{zhou2017differential}. In this case, we say $\rho\sim\sigma$ if and only if $\mathcal{T}(\rho,\sigma)\leq d$ for some constant $d>0$. However, we may select another notion of similarity that ensures that two quantum states are neighbouring if they are constructed based on two private datasets that differ in the data of one individual. Such a definition is well-suited for quantum machine learning with privacy guarantees~\cite{watkins2021quantum}.

\begin{definition}[$(\varepsilon,\eta)$-Privacy Against Hypothesis-Testing Adversary] For any $\varepsilon\geq 0$ and $\eta\in[0,1]$, a quantum channel $\mathcal{E}$ is $(\varepsilon,\eta)$-private (against hypothesis-testing adversary) if $D^{\eta }(\mathcal{E}(\rho) \|\mathcal{E}(\sigma) )\leq \varepsilon$ for all neighbouring states $\rho\sim\sigma$.
\end{definition}

This definition implies that, if two states are similar $\rho\sim\sigma$, a quantum channel $\mathcal{E}$ is private if it makes distinguishing the reported or output states $\mathcal{E}(\rho)$ and $\mathcal{E}(\sigma)$ difficult by \textit{any} decision maker. In fact, Proposition~\ref{prop:assym} shows that probability of false negatives $\beta(M_1)$ for any detection mechanism $M=\{M_1,M_2\}$ is lower bounded by $2^{-\epsilon}$ if the probability of false positives bounded by $\alpha(M_2)\leq \eta$. Therefore, as $\epsilon$ tends to zero (privacy guarantee is strengthened/privacy budget is reduced), the probability of false negatives move towards one (i.e., the decision maker would become overwhelmed by false negatives). 

\begin{proposition} \label{prop:change_parameters}
    Assume that a quantum channel $\mathcal{E}$ is $(\varepsilon,\eta)$-private. Then, the quantum channel $\mathcal{E}$ is $(\varepsilon',\eta')$-private if $\eta'\geq \eta$ and $\varepsilon\leq \varepsilon'$.
\end{proposition}

\begin{proof}
    First note that, if $\eta'\geq \eta$, we have
    \begin{align*}
        2^{-D^{\eta }(\rho,\sigma)}
        &=\min\{\trace(Q\sigma) |0\!\preceq \!Q\!\preceq \!I, \trace(Q\rho) \geq 1\!-\!\eta \}\\
        &\leq \min\{\trace(Q\sigma) |0\!\preceq \!Q\!\preceq \!I, \trace(Q\rho) \geq 1\!-\!\eta' \}\\
        &=2^{-D^{\eta' }(\rho,\sigma)},
    \end{align*}
where the inequality follows from that $\{Q|0\!\preceq \!Q\!\preceq \!I, \trace(Q\rho) \geq 1\!-\!\eta\}\subseteq \{Q|0\!\preceq \!Q\!\preceq \!I, \trace(Q\rho) \geq 1\!-\!\eta'\}$. Therefore, for all $\sigma\sim\rho$, we get $D^{\eta' }(\rho,\sigma)\leq D^{\eta }(\rho,\sigma) \leq \epsilon \leq \epsilon'.$
\end{proof}

The following corollary, building on Proposition~\ref{prop:change_parameters}, shows that $(\varepsilon,0)$-privacy against hypothesis testing adversary is the strongest notion of privacy and thus, $(\varepsilon,\eta)$-privacy can be thought of as relaxations of $(\varepsilon,0)$-privacy.

\begin{corollary}
Assume that a quantum channel $\mathcal{E}$ is $(\varepsilon,0)$-private. Then, the quantum channel $\mathcal{E}$ is $(\varepsilon,\eta)$-private for all $\eta\in[0,1]$. 
\end{corollary}

Although privacy is here defined in terms of asymmetric quantum hypothesis testing, we prove the following important bound on the power of symmetric quantum hypothesis testing. 

\begin{theorem} \label{tho:assym_to_symm} For any $(\varepsilon,\eta)$-private quantum channel $\mathcal{E}$, 
\begin{align} \label{eqn:p_err_lower_bound}
    p_{\rm err}(\mathcal{E}(\rho),\mathcal{E}(\sigma))\geq \Gamma_{p_\rho,p_\sigma}(\varepsilon,\eta),
\end{align}
where
\begin{align*}
    \Gamma_{p_\rho,p_\sigma}(\varepsilon,\eta)
    := &\max\left\{p_{\rm max}\!-\!\frac{\varepsilon \min\{p_\rho,p_\sigma\}(1\!-\!\eta)}{2\eta},0\right\}.
\end{align*}
\end{theorem}

\begin{proof}
    Using~\cite{dupuis2014generalized}, we have
    \begin{align*}
        \frac{\eta}{1-\eta}\|\mathcal{E}(\rho)-\mathcal{E}(\sigma)\|_1\leq D^{\eta }(\mathcal{E}(\rho) \|\mathcal{E}(\sigma) ).
    \end{align*}
    Therefore, if $\mathcal{E}$ is $(\varepsilon,\eta)$-private (against hypothesis-testing adversary), we get
    \begin{align*}
        \|\mathcal{E}(\rho)-\mathcal{E}(\sigma)\|_1\leq \frac{1-\eta}{\eta}\varepsilon.
    \end{align*}
We have
\begin{align}
\|p_\rho\mathcal{E}(\rho)\!-\!p_\sigma\mathcal{E}(\sigma)\|_1
\!=&p_\sigma\left\|\frac{p_\rho}{p_\sigma}\mathcal{E}(\rho) \!-\!\mathcal{E}(\sigma)\right\|_1
\nonumber\\
=&p_\sigma\left\|\frac{p_\rho\!-\!p_\sigma}{p_\sigma}\mathcal{E}(\rho) \!+\!\mathcal{E}(\rho) \!-\!\mathcal{E}(\sigma)\right\|_1
\nonumber\\
\leq & |p_\rho\!-\!p_\sigma|\|\mathcal{E}(\rho) \|_1\!+\!p_\sigma \|\mathcal{E}(\rho) \!-\!\mathcal{E}(\sigma)\|_1\nonumber\\
\leq & |p_\rho\!-\!p_\sigma|\!+\!\frac{1-\eta}{\eta}\varepsilon p_\sigma.
\label{eqn:proof:3}
\end{align}
Following the same line of reasoning, we can also show that
\begin{align}
\|p_\rho\mathcal{E}(\rho)-p_\sigma\mathcal{E}(\sigma)\|_1
\leq  |p_\rho-p_\sigma|+\frac{1-\eta}{\eta}\varepsilon p_\rho.\label{eqn:proof:4}
\end{align}
Combining~\eqref{eqn:proof:3} and~\eqref{eqn:proof:4}, we get
\begin{align*}
\|p_\rho\mathcal{E}(\rho)-p_\sigma\mathcal{E}(\sigma)\|_1
\leq  |p_\rho-p_\sigma|+\frac{1-\eta}{\eta}\varepsilon \min\{p_\rho,p_\sigma\}.
\end{align*}
Therefore, 
\begin{align*}
    p_{\rm err}(\mathcal{E}(\rho),\mathcal{E}(\sigma))
    =&p_{\rm max}
    +\frac{1}{2}\Big(|p_\rho-p_\sigma|\\
    &\hspace{.7in}-\|p_\rho\mathcal{E}(\rho)-p_\sigma\mathcal{E}(\sigma)\|_1\Big)\\
    \geq& p_{\rm max}-\frac{\varepsilon \min\{p_\rho,p_\sigma\}(1-\eta)}{2\eta}.
\end{align*}
This concludes the proof.
\end{proof}

Theorem~\ref{tho:assym_to_symm} shows that, by decreasing $\varepsilon$, the combined probabilities of false positive and false negative denoted by $p_{\rm err}(\mathcal{E}(\rho),\mathcal{E}(\sigma))$ increases towards its maximum value $p_{\max}$. Figure~\ref{fig:assymetric_privacy} illustrates the lower bound $\Gamma_{p_\rho,p_\sigma}(\varepsilon,\eta)$ on $p_{\rm err}(\mathcal{E}(\rho),\mathcal{E}(\sigma))$ versus the privacy budget $\varepsilon$ for various choices of $\eta$ for the case that $p_\rho=p_\sigma=\frac{1}{2}$. As expected, reducing the privacy budget $\varepsilon$ strengthens the privacy guarantees. 

\begin{figure}
    \centering
    \begin{tikzpicture}
    \node[] at (0,0) {\includegraphics[width=.95\linewidth]{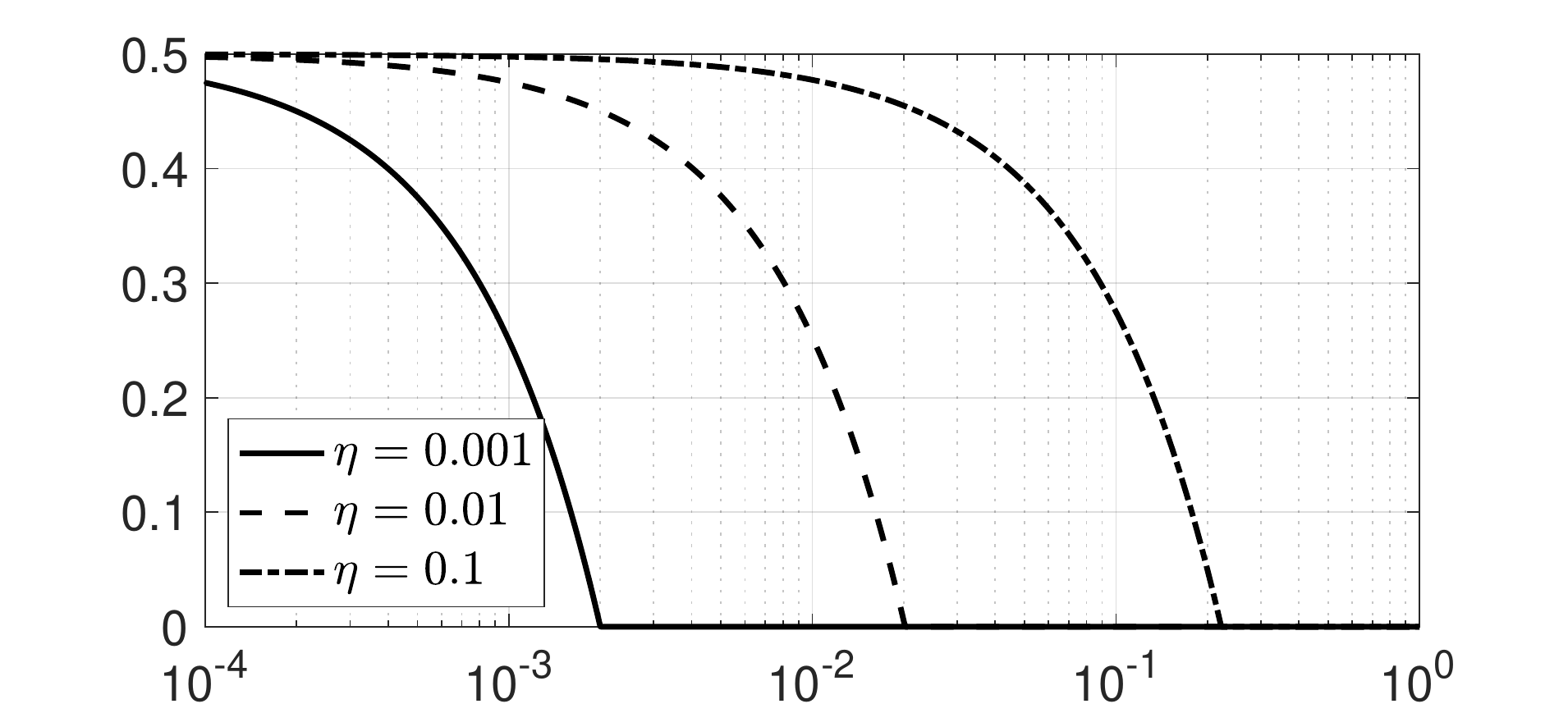}};
    \node[rotate=90] at (-3.8,0) {\small $\Gamma_{\frac{1}{2},\frac{1}{2}}(\varepsilon,\eta)$ in Theorem~\ref{tho:assym_to_symm} };
    \node[] at (0,-2.1) {$\varepsilon$};
    \end{tikzpicture}
    \vspace{-3mm}
    \caption{The lower bound $\Gamma_{\frac{1}{2},\frac{1}{2}}(\varepsilon,\eta)$ on $p_{\rm err}(\mathcal{E}(\rho),\mathcal{E}(\sigma))$ in Theorem~\ref{tho:assym_to_symm} versus the privacy budget $\varepsilon$ for various choices of $\eta$. As expected, reducing the privacy budget $\varepsilon$ strengthens the privacy guarantees. }
    \label{fig:assymetric_privacy}
\end{figure}

It is stipulated that any useful notion of privacy should admit two important properties of post processing and composition~\cite{hirche2022quantum}. In the remainder of this section, we discuss these properties and their application to privacy against hypothesis testing adversary.

\begin{theorem}[Post Processing] \label{tho:post_processing} Let $\mathcal{E}$ be any $(\varepsilon,\eta)$-private and $\mathcal{N}$ be an arbitrary quantum channel, then $\mathcal{N}\circ\mathcal{E}$ is $(\varepsilon,\eta)$-private. 
\end{theorem}

\begin{proof}
    The proof follows from that, for all $\rho$ and $\sigma$, $D^{\eta }(\mathcal{N}(\mathcal{E}(\rho)) \|\mathcal{N}(\mathcal{E}(\sigma)) )\leq D^{\eta }(\mathcal{E}((\rho) \|\mathcal{E}(\sigma) )$~\cite{wang2012one}.
\end{proof}

Theorem~\ref{tho:post_processing} shows that an adversary cannot weaken the privacy guarantees by processing the received quantum information in any way. 

\begin{theorem}[Composition] \label{tho:composition} Let $\mathcal{E}_1$ be any $(\varepsilon_1,0)$-private and $\mathcal{E}_2$ be any $(\varepsilon_2,0)$-private. Assume that $\rho_1\otimes\rho_2\sim \sigma_1\otimes\sigma_2$ if $\rho_1\sim \sigma_1$ and $\rho_2\sim \sigma_2$. Then,  Then $\mathcal{E}_1\otimes\mathcal{E}_2$ is $(\varepsilon_1+\varepsilon_2,0)$-private.
\end{theorem}

\begin{proof}
Using the additivity results in~\cite[Appendix~A]{yuan2019hypothesis}, we get $D^{0}(\rho_1\otimes\rho_2 \|\sigma_1\otimes\sigma_2 )
    =D^{0}(\rho_1\|\sigma_1)+D^{0}(\rho_2\|\sigma_2).$ Therefore, if $D^{0}(\rho_i\|\sigma_i)\leq \varepsilon_i$ for $i=1,2$, then $D^{0}(\rho_1\otimes\rho_2 \|\sigma_1\otimes\sigma_2 )\leq \varepsilon_1+\varepsilon_2$. 
\end{proof}

In practical data processing applications, there is often a need to deal with complicated algorithms in which responses from several queries based on private user data are fused together to extract useful statistical information from the data. For instance, when training machine learning models, iterative gradient descent algorithms can be used and the gradient at each epoch can be modelled as a query on the private data used for training~\cite{wu2020value}. In this case, it is desirable to establish composition rules for combination of several privacy-preserving quantum operations. Theorem~\ref{tho:composition} provides such a result for privacy against hypothesis testing adversaries. 

\section{Quantum Differential Privacy} \label{sec:diff_privacy}
The gold standard of privacy analysis and enforcement in the computer science literature is differential privacy, which has been recently extended to quantum computing algorithms~\cite{zhou2017differential}. In this section, we establish a relationship between differential privacy and privacy against hypothesis testing adversaries. 

\begin{definition}
    For any $\epsilon,\delta\geq 0$, a quantum channel $\mathcal{E}$ is $(\epsilon,\delta)$-differentially private if 
    \begin{align}
        \trace(M\mathcal{E}(\rho))\leq \exp(\epsilon) \trace(M\mathcal{E}(\sigma))+\delta,
    \end{align}
    for all measurements $0\preceq M\preceq I$ and neighbouring density operators $\rho\sim\sigma$. 
\end{definition}

We can prove the following result regarding the relationship between quantum differential privacy and privacy against hypothesis testing adversaries. 

\begin{theorem} \label{tho:equivalence} The following two statements hold:
\begin{itemize}
    \item If $\mathcal{E}$ be $(\varepsilon,\eta)$-private, then $\mathcal{E}$ is $(\varepsilon,\sqrt{2\eta})$-differentially private.
    \item If $\mathcal{E}$ be $(\epsilon,0)$-differentially private, then $\mathcal{E}$ is $(\varepsilon,\eta)$-private for all $\eta\in[0,1]$.
\end{itemize}
    
\end{theorem}

\begin{proof} First, $D_{\max}^{\sqrt{2\nu}}(\mathcal{E}(\rho),\mathcal{E}(\sigma))\leq D^\nu(\mathcal{E}(\rho),\mathcal{E}(\sigma))$~\cite[Proposition~4.1]{dupuis2014generalized}. Therefore, if $\mathcal{E}$ is $(\varepsilon,\eta)$-private, we get $D_{\max}^{\sqrt{2\eta}}(\mathcal{E}(\rho),\mathcal{E}(\sigma))\leq D^\eta(\mathcal{E}(\rho),\mathcal{E}(\sigma))\leq \varepsilon$. From Lemma~III.2 in~\cite{hirche2022quantum}, a quantum channel $\mathcal{E}$ is $(\epsilon,\delta)$-differentially private if and only if $D_{\max}^\delta(\mathcal{E}(\rho),\mathcal{E}(\sigma))\leq \epsilon$. This proves that $\mathcal{E}$ is $(\varepsilon,\sqrt{2\eta})$-differentially private.

For the second part, note that $D^\eta(\mathcal{E}(\rho),\mathcal{E}(\sigma))\leq D_{\max}^0(\mathcal{E}(\rho),\mathcal{E}(\sigma))$~\cite[Proposition~4.1]{dupuis2014generalized}. Therefore, if $\mathcal{E}$ is $(\epsilon,0)$-differentially private, we have $D^\eta(\mathcal{E}(\rho),\mathcal{E}(\sigma))\leq D_{\max}^0(\mathcal{E}(\rho),\mathcal{E}(\sigma))\leq \epsilon$. This implies that $\mathcal{E}$ is $(\epsilon,\eta)$-private for all $\eta\in[0,1]$.
\end{proof}

\begin{theorem}[Lemma~IV.2~\cite{hirche2022quantum}] \label{tho:DP_DC}
    Consider neighbourhood notion that $\rho\sim\sigma$ if $\mathcal{T}(\rho,\sigma)\leq d$. Then, the depolarizing channel $\mathcal{E}_{\rm Dep}(\rho)$ is $(\epsilon,\delta)$-differentially private with $\delta=\max\{0,(1-\exp(\epsilon))p/D+(1-p)\kappa\}.$
\end{theorem}

\begin{corollary}
    Consider neighbourhood notion that $\rho\sim\sigma$ if $\mathcal{T}(\rho,\sigma)\leq d$. Then, the depolarizing channel $\mathcal{E}_{\rm Dep}(\rho)$ is $(\varepsilon,\eta)$-private with $\varepsilon=\log(1+(1-p)D\kappa/p)$ and all $\eta\in[0,1]$.
\end{corollary}

\begin{proof}
    First, note that Theorem~\ref{tho:DP_DC} shows that the depolarizing channel $\mathcal{E}_{\rm Dep}(\rho)$ is $(\epsilon,\delta)$-differentially private with $\delta=\max\{0,(1-\exp(\epsilon))p/D+(1-p)\kappa\}.$ If we select $\epsilon=\log(1+(1-p)D\kappa/p)$, we get $\delta=0$. Using Theorem~\ref{tho:equivalence}, then $\mathcal{E}$ is $(\epsilon,\eta)$-private for all $\eta\in[0,1]$.
\end{proof}

We finish this section with analysing the performance of hypothesis testing adversaries for differentially-private quantum channels. 

\begin{theorem} \label{tho:HT_AS_DP}
For any $(\epsilon,\delta)$-differentially private quantum channel $\mathcal{E}$,
    \begin{align}\label{eqn:lower_bound_beta}
    \beta_\eta(\mathcal{E}(\rho),\mathcal{E}(\sigma))\geq \Omega_\eta(\epsilon,\delta),
\end{align}
where $\Omega_\eta(\epsilon,\delta):=\exp(-\epsilon)(1-\eta-\delta).$
\end{theorem}

\begin{proof}
Assume that $\rho\sim\sigma$. Because of $(\epsilon,\delta)$-differential privacy, $\trace(M\mathcal{E}(\sigma))\geq \exp(-\epsilon)(\trace(M\mathcal{E}(\rho))-\delta)$ for all measurements $0\preceq M\preceq I$.  Therefore, 
    $\beta_\eta(\mathcal{E}(\rho),\mathcal{E}(\sigma)) 
        \!=\!\min_{I \succeq M\succeq 0} \{\trace(M\mathcal{E}(\sigma))|
            1\!-\!\eta\leq \trace(M\mathcal{E}(\rho))\}
            \geq \! \exp(-\epsilon)(1-\eta-\delta).$
\end{proof}

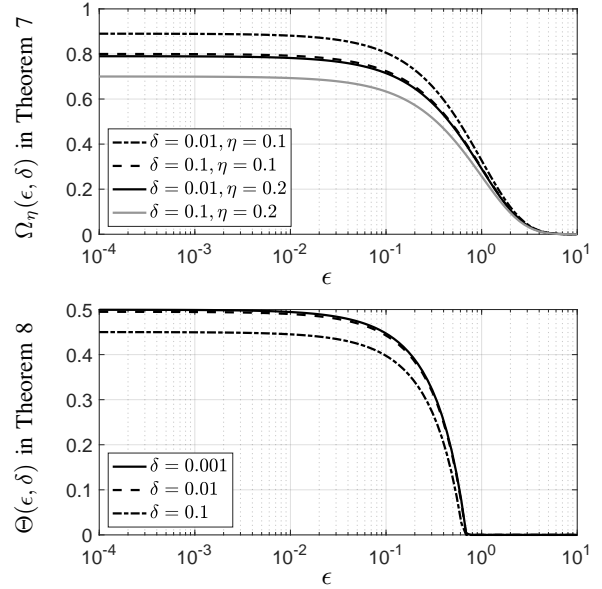
\begin{figure}
    \centering
    \begin{tikzpicture}
    \node[] at (0,0) {\includegraphics[width=.95\linewidth]{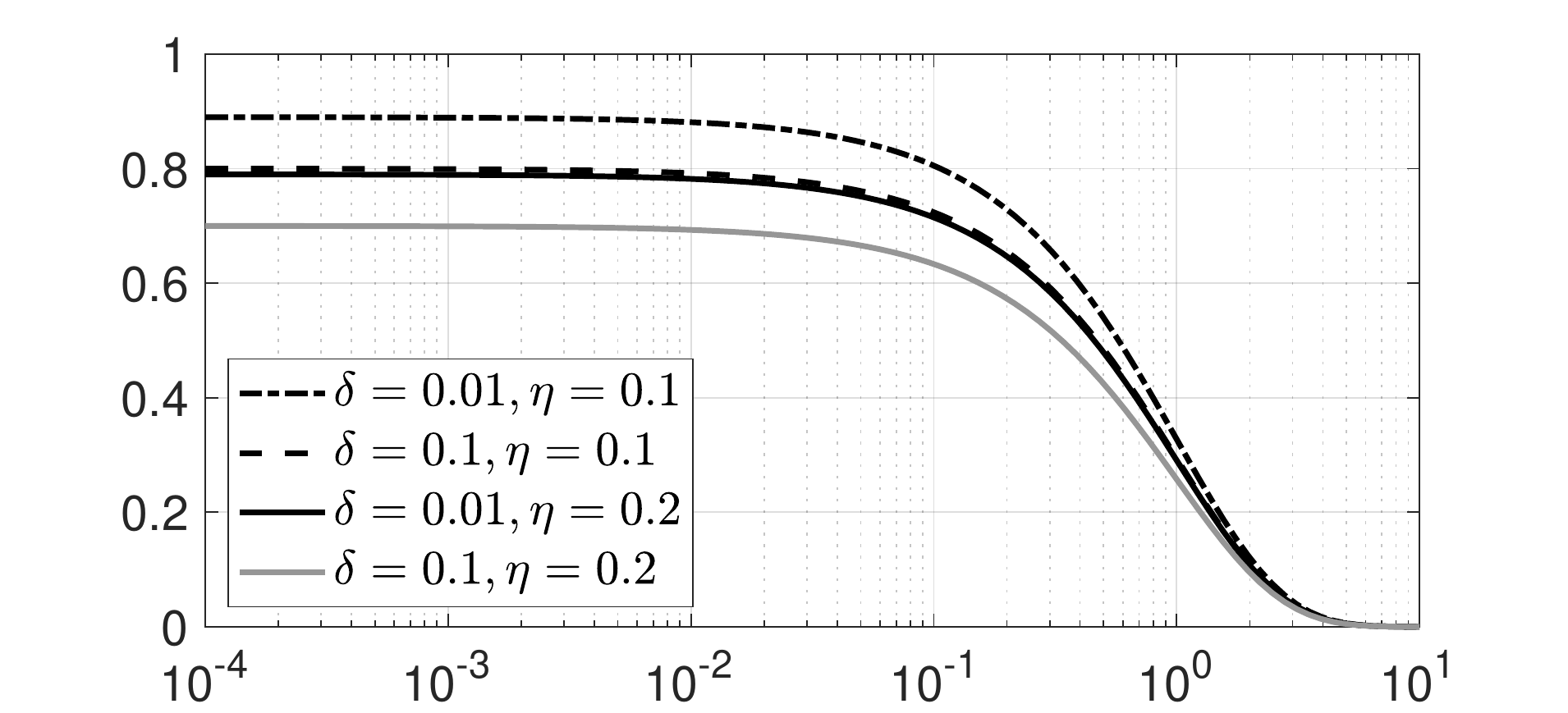}};
    \node[rotate=90] at (-4,0) {\small $\Omega_\eta(\epsilon,\delta)$ in Theorem~\ref{tho:HT_AS_DP}};
    \node[] at (0,-4) {\includegraphics[width=.95\linewidth]{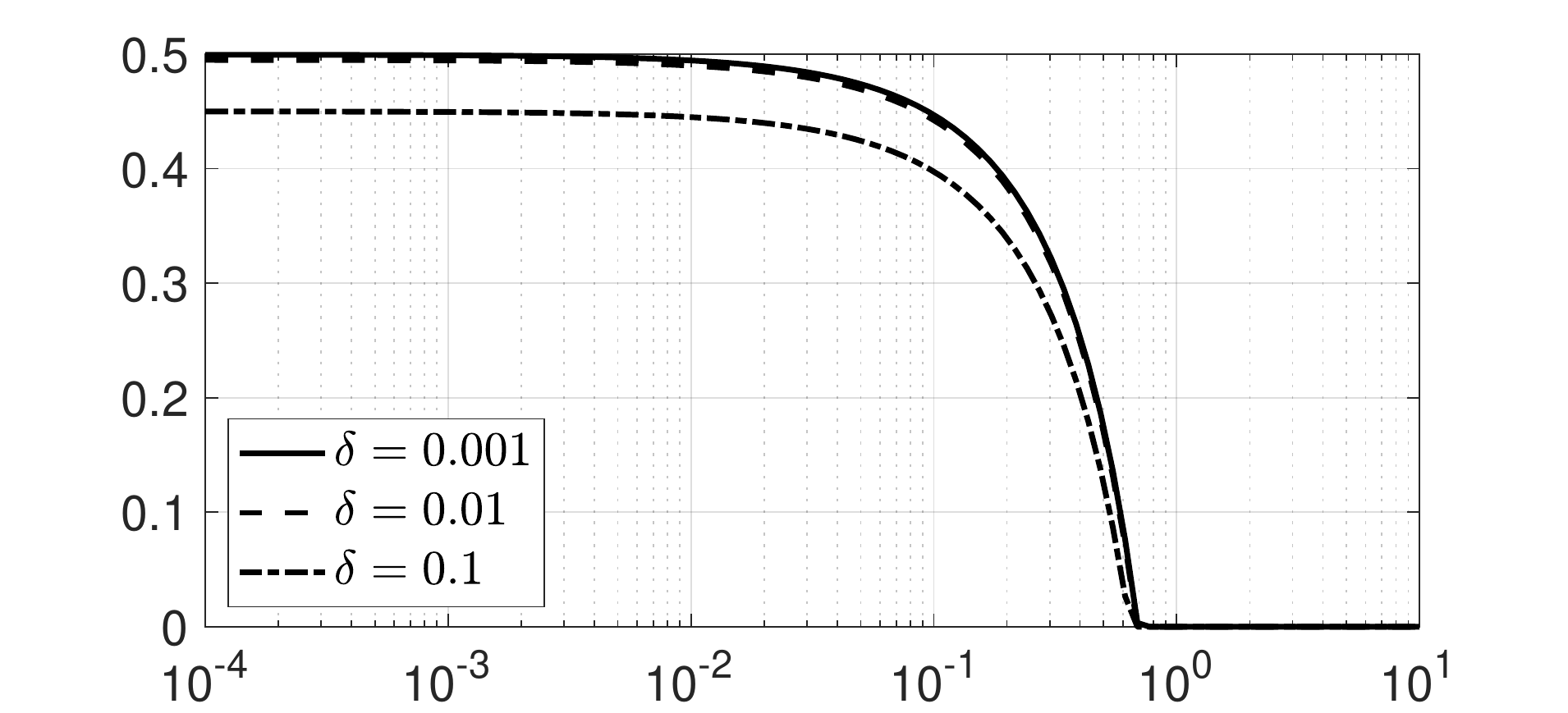}};
    \node[rotate=90] at (-4,-4) {\small $\Theta(\epsilon,\delta)$ in Theorem~\ref{tho:HT_S_DP}};
    \node[] at (0,-2) {$\epsilon$};
    \node[] at (0,-6) {$\epsilon$};
    \end{tikzpicture}
    \caption{Lower bound on $\beta_\eta(\mathcal{E}(\rho),\mathcal{E}(\sigma))$ in Theorem~\ref{tho:HT_AS_DP} [top] and lower bound on $p_{\rm err}(\mathcal{E}(\rho),\mathcal{E}(\sigma))$ in Theorem~\ref{tho:HT_S_DP} [bottom] versus privacy budget $\epsilon$ for various choices of $\delta$. }
    \label{fig:1}
\end{figure}

Theorem~\ref{tho:HT_AS_DP} provides a lower bound for the false negative rate for the best asymmetric hypothesis testing mechanism. The lower bound grows, and thus the decision maker would get overwhelmed by false negatives, with decreasing $\epsilon$ and $\delta$. Therefore, the privacy guarantees strengthens as the privacy budget reduces in quantum differential privacy. This is illustrated in Figure~\ref{fig:1} [top].

\begin{theorem}\label{tho:HT_S_DP}
For any $(\epsilon,\delta)$-differentially private quantum channel $\mathcal{E}$,
    \begin{align} 
    p_{\rm err}(\mathcal{E}(\rho),\mathcal{E}(\sigma))\!\geq &\Theta(\epsilon,\delta),
    \label{eqn:lower_bound_p_error}
\end{align}
where $\Theta(\epsilon,\delta):=\max\{p_{\rm max}+\max\{p_\rho,p_\sigma\}(1-\exp(\epsilon)-\delta),0\}.$
\end{theorem}

\begin{proof} 
First, assume that $p_\rho\geq p_\sigma$. The definition of differential privacy implies that $\trace(\Lambda\mathcal{E}(\rho))\leq \exp(\epsilon)\trace(\Lambda\mathcal{E}(\sigma))+\delta$ for all $0\preceq \Lambda \preceq I$. As a result,
\begin{align*}
    \trace(\Lambda(p_\rho\mathcal{E}(\rho)-p_\sigma\mathcal{E}(\sigma)))
    \leq & (p_\rho\exp(\epsilon)\!-\!p_\sigma)\trace(\Lambda\mathcal{E}(\sigma))\!+\!p_\rho\delta\\
    \leq & p_\rho\exp(\epsilon)-p_\sigma+p_\rho\delta,
    \end{align*}
where the last inequality follows from that $p_\rho\exp(\epsilon)-p_\sigma\geq p_\rho-p_\sigma\geq 0$ and that $\trace(\Lambda\mathcal{E}(\sigma))\leq 1$ because $0\preceq \Lambda \preceq I$. Therefore, using Lemma~\ref{lemma:appendix} in the appendix, we have
\begin{align}
    \frac{1}{2}\|p_\rho\mathcal{E}(\rho)-p_\sigma\mathcal{E}(\sigma)\|_1
    \leq& p_\rho\exp(\epsilon)-p_\sigma+p_\rho\delta+\frac{p_\sigma-p_\rho}{2}\nonumber\\
    \leq& p_\rho(\exp(\epsilon)+\delta)-\frac{p_\sigma+p_\rho}{2}.
    \label{eqn:proof:1}
\end{align}
Alternatively, assume that $p_\sigma\geq p_\rho$. Following the same line of reasoning, we get
\begin{align}
    \frac{1}{2}\|p_\sigma\mathcal{E}(\sigma)-p_\rho\mathcal{E}(\rho)\|_1
    \leq& p_\sigma(\exp(\epsilon)+\delta)-\frac{p_\sigma+p_\rho}{2}.
    \label{eqn:proof:2}
\end{align}
Combining~\eqref{eqn:proof:1} and~\eqref{eqn:proof:2} gives
\begin{align*}
    \frac{1}{2}\|p_\sigma\mathcal{E}(\sigma)-p_\rho\mathcal{E}(\rho)\|_1
    \leq& \max\{p_\rho,p_\sigma\}(\exp(\epsilon)+\delta)\\
    &-\frac{p_\sigma+p_\rho}{2}.
\end{align*}
Therefore, 
\begin{align*}
    |p_\rho-&p_\sigma|-\|p_\rho\mathcal{E}(\rho)-p_\sigma\mathcal{E}(\sigma)\|_1\\
    \geq& |p_\rho-p_\sigma|+(p_\sigma+p_\rho)-2\max\{p_\rho,p_\sigma\}(\exp(\epsilon)+\delta)\\
    =&2\max\{p_\rho,p_\sigma\}(1-\exp(\epsilon)-\delta).
\end{align*}
This concludes the proof.
\end{proof}

Theorem~\ref{tho:HT_S_DP} provides a lower bound for the combined false positive and negative rates of the best symmetric hypothesis testing mechanism. The lower bound grows towards $p_{\max}$ as $\epsilon$ and $\delta$ become smaller, which demonstrates that the privacy guarantees strengthen as the privacy budget reduces in quantum differential privacy. This is illustrated in Figure~\ref{fig:1} [bottom].

\section{Conclusions and Future Work} \label{sec:discussions}
We presented a novel definition for privacy in quantum computing based on quantum hypothesis testing. Important properties of post processing and composition were proved for this new notion of privacy. We then examined the relationship between privacy against  hypothesis-testing adversaries, defined in this paper, and quantum differential privacy are then examined. In the composition rules for privacy against hypothesis adversaries, we only considered the case of $\eta=0$. Future work can expand these results for general case of $\eta\in[0,1]$. Furthermore, we only showed that $(\epsilon,0)$-differential privacy can be translated to privacy against hypothesis testing adversaries (the inverse results are more general in this paper). Therefore, another avenue for future research is to expand these results to general $(\epsilon,\delta)$-differential privacy. Finally, an important direction for future research is to use the proposed framework in numerical setups based on real-world data.

\bibliography{ref}
\bibliographystyle{ieeetr}

\appendix

\begin{lemma} \label{lemma:appendix} The following identity holds:
\begin{align*}
    \frac{1}{2}\|p_\rho\rho-p_\sigma\sigma\|_1=\max_{0\preceq \Lambda\preceq I} \trace(\Lambda(p_\rho\rho-p_\sigma\sigma))+\frac{p_\sigma-p_\rho}{2}.
\end{align*}
\end{lemma}

\begin{proof}
    The proof is similar to the standard argument for the trace distance. 
    Note that the difference operator $p_\rho\rho-p_\sigma\sigma$ is Hermitian. So we can diagonalize it as $p_\rho\rho-p_\sigma\sigma=\sum_i \lambda_i \ket{i}\bra{i}$, where $\{\ket{i}\}_i$ is an orthonormal basis of eigenvectors and $\{\lambda_i\}_i$ is a set of real eigenvalues. Define matrices $P:=\sum_{i:\lambda_i>0} \lambda_i \ket{i}\bra{i}\succeq 0$ and $Q:=\sum_{i:\lambda_i<0} (-\lambda_i)\ket{i}\bra{i}\succeq 0$. Evidently, by construction, $p_\rho\rho-p_\sigma\sigma=P-Q$. Note that, 
\begin{align*}
    \|p_\rho\rho-p_\sigma\sigma\|_1
    =&\trace(|p_\rho\rho-p_\sigma\sigma|)\\
    =&\trace(|P-Q|)\\
    =&\trace(P+Q)\\
    =&2\trace(P)+(p_\sigma-p_\rho),
\end{align*}
where the last equality follows from
\begin{align*}
    \trace(P)-\trace(Q)
    =&\trace(P-Q)\\
    =&\trace(p_\rho\rho-p_\sigma\sigma)\\
    =&p_\rho\trace(\rho)
    -p_\sigma\trace(\sigma)\\
    =&p_\rho-p_\sigma.
\end{align*}
For all $0\preceq \Lambda\preceq I$, we have
    \begin{align*}
        \trace(\Lambda(p_\rho\rho-p_\sigma\sigma))
        =&\trace(\Lambda(P-Q))\\
        \leq & \trace(\Lambda P)\\
        \leq & \trace(P)\\
        =&\frac{1}{2}\|p_\rho\rho-p_\sigma\sigma\|_1
        +\frac{p_\rho-p_\sigma}{2},
    \end{align*}
with equality achieved if $P=\sum_{i:\lambda_i>0} \ket{i}\bra{i}$. This implies that
\begin{align*}
    \frac{1}{2}\|p_\rho\rho-p_\sigma\sigma\|_1=\max_{0\preceq \Lambda\preceq I} \trace(\Lambda(p_\rho\rho-p_\sigma\sigma))+\frac{p_\sigma-p_\rho}{2}.
\end{align*}
\end{proof}

\end{document}